\documentclass[english,conference,letter]{IEEEtran}
\usepackage[T1]{fontenc}
\usepackage{url}
\usepackage{ifthen}
\usepackage{cite}
\usepackage{tikz}
\usetikzlibrary{positioning}
\usetikzlibrary{arrows,arrows.meta}
\tikzstyle{block}=[draw opacity=0.7,line width=1.4cm]
\usetikzlibrary{positioning,shapes,shadows,arrows,scopes,decorations}
\usetikzlibrary{matrix,chains,shapes.geometric,shapes.arrows,automata,chains,er,fit}

\tikzstyle{comment}=[rectangle, draw=black, fill=red!50!black, rounded corners, drop shadow,
anchor=west, text=white]
\usetikzlibrary{shapes.callouts}
\usetikzlibrary{mindmap}
\usetikzlibrary{calc,through,backgrounds}

\usepackage{graphicx}
\usepackage{mathtools}
\usepackage{amsmath,amssymb}
\usepackage{cite}               
\usepackage{enumerate}
\usepackage{amssymb}
\usepackage{amsthm}
\usepackage{float}
\usepackage{bbm}
\usepackage{caption}
\usepackage[utf8]{inputenc}
\usepackage[english]{babel}
\usepackage{dsfont}	
\newtheorem{theorem}{Theorem}
\newtheorem{claim}{Claim}

\newtheorem{corollary}{Corollary}
\newtheorem{lemma}{Lemma}

\theoremstyle{remark}
\def\Bern{\mathsf{Bern}}
\def\Mat{\mathbf{P}}

\DeclareMathOperator{\E}{\mathsf{E}}
\DeclareMathOperator{\Pe}{\mathsf{P_e}}
\DeclareMathOperator{\Pstar}{\mathsf{P^*_e}}
\DeclareMathOperator{\ind}{\mathds{1}}

\DeclareMathOperator{\rand}{rand}

\DeclareMathOperator{\Pm}{\mathsf{P_{\min}}}
\DeclareMathOperator{\td}{\mathsf{td}}
\DeclareMathOperator{\occ}{\mathsf{occ}}

\IEEEoverridecommandlockouts

\begin{document}

\title{Binary Hypothesis Testing with Deterministic Finite-Memory Decision Rules
}

\author{\IEEEauthorblockN{Tomer Berg}
\IEEEauthorblockA{Tel Aviv University\\ tomerberg@mail.tau.ac.il}
\and
\IEEEauthorblockN{Or Ordentlich}
\IEEEauthorblockA{Hebrew University of Jerusalem\\ or.ordentlich@mail.huji.ac.il}
\and
\IEEEauthorblockN{Ofer Shayevitz}
\IEEEauthorblockA{Tel Aviv University\\ ofersha@eng.tau.ac.il}
\IEEEoverridecommandlockouts
\IEEEcompsocitemizethanks{
\IEEEcompsocthanksitem
The work of Tomer Berg was supported by the ISF under Grant 1791/17 and the ERC under Grant 639573. The work of Or Ordentlich was supported by the ISF under Grant 1791/17. The work of Ofer Shayevitz was supported by the ERC under Grant 639573.
}
}

\maketitle

\begin{abstract} 
	In this paper we consider the problem of binary hypothesis testing with finite memory systems. Let $X_1,X_2,\ldots$ be a sequence of independent identically distributed Bernoulli random variables, with expectation $p$ under $\mathcal{H}_0$ and $q$ under $\mathcal{H}_1$. Consider a finite-memory deterministic machine with $S$ states that updates its state $M_n \in \{1,2,\ldots,S\}$ at each time according to the rule $M_n = f(M_{n-1},X_n)$, where $f$ is a deterministic time-invariant function. Assume that we let the process run for a very long time ($n\rightarrow \infty)$, and then make our decision according to some mapping from the state space to the hypothesis space. 
	The main contribution of this paper is a lower bound on the Bayes error probability $P_e$ of any such machine. 
	In particular, our findings show that the ratio between the maximal exponential decay rate of $P_e$ with $S$ for a deterministic machine and for a randomized one, can become unbounded, complementing a result by Hellman.

\end{abstract}

	\section{Introduction}
	Consider the following binary hypothesis testing problem: $X_1,X_2,\ldots$ is a sequence of independent identically distributed random variables drawn according to either the $\Bern(p)$ distribution, under hypothesis $\mathcal{H}_0$, or the $\Bern(q)$ distribution, under hypothesis $\mathcal{H}_1$, for $0<q<p<1$. For simplicity, we assume throughout that the prior probabilities of both hypothesis are given and are equal. A finite memory decision rule for this problem is a triplet $(S,f,d)$ where $S$ is the number of states used by the machine, $f:[S] \times \{0,1\}\rightarrow [S]$ is the state transition function, and $d:[S]\rightarrow \{\mathcal{H}_0,\mathcal{H}_1\}$ is the decision function. In contrast to much of the prior work,  where randomized state-transition functions $f$ were allowed, here we restrict our attention to \emph{deterministic} $f$.
	
	Letting $M_n$ denote the state of the memory at time $n$, the finite state machine evolves according to the rule
	\begin{align}
	    M_0&=s, \\M_n&=f(M_{n-1},X_n)\in [S],
	\end{align}
	for some $s \in [S]$. If the machine is stopped at time $n$, it outputs the decision $d(M_n)$.
	
	Conditioned on $\mathcal{H}_0$, the  process $\{M_n\}$, induced by the function $f$, is a Markov chain with stochastic transition matrix 
	\begin{align}
    \Mat (p) = \left[\Pr\left(f(i,X)=j|\mathcal{H}_0\right)\right]=[p_{ij}(p)]  ,
	\end{align}
	for all $i,j\in [S]$. Similarly, under $\mathcal{H}_1$, the induced Markov chain has stochastic transition matrix  $\Mat (q) = \left[\Pr\left(f(i,X)=j|\mathcal{H}_1\right)\right]=[p_{ij}(q)]$. 
Following~\cite{hellman1970learning}, we define the asymptotic probability of error of an algorithm as
	\begin{align}
	    \Pe (S,f,d) = \lim_{n\rightarrow\infty} \frac{1}{n}\sum_{i=1}^n \Pr (e_i=1),\label{eq:pefd} 
	\end{align}
	where $e_i=\ind_{\{d(M_i)\neq \mathcal{H}_t\}}$, 
	and $\mathcal{H}_t$ is the true hypothesis. Arguably, a more natural definition of error probability is
	\begin{align}
	    \Pe (S,f,d) = \limsup_{n\rightarrow\infty} \Pr \left(d(M_n)\neq \mathcal{H}_t\right) . \label{eq:stringent}
	\end{align}
	However, as~\eqref{eq:stringent} is always larger than~\eqref{eq:pefd}, by a factor of at most $S$, the two definitions are equivalent for the purposes of this study.

The focus of this paper is the quantity
\begin{align}
	    \mathsf{P^*_e}(S)=\min_{\text{deterministic} f,d} \Pe (f,d)
	\end{align} 
	 where the minimum is taken over all $S$-state machines with \textit{deterministic} transition functions $f$. We are specifically interested in the asymptotics of the error exponent with regards to $S$,
	\begin{align}
	 \overline{\E} (p,q)&=-\liminf_{S \rightarrow \infty} \frac{1}{S}\log \Pstar(S),\\ \underline{\E} (p,q)&=-\limsup_{S \rightarrow \infty} \frac{1}{S}\log \Pstar(S).\label{eq:Eunder_pq}
	\end{align}

\subsection{Related work}
It seems that interest in the limited memory binary hypothesis testing problem was sparked by the work of Robbins~\cite{robbins1956sequential} on the Two-Armed Bandit problem: A player is given two coins, with parameters unknown to him, and is required to maximize the long-run proportions of "heads" obtained, by successively choosing which coin to flip at any moment. Robbins proposed an algorithm that works with limited memory $S$. Cover~\cite{cover1968note} discovered a time-varying finite memory algorithm that achieves the maximum with $S=2$, and in a subsequent paper addressing the binary hypothesis problem~\cite{cover1969hypothesis} described a time-varying finite memory machine that has probability of error approaching zero with $S=4$.  Due to the unlimited memory that is needed to implement a time-varying machine, Hellman and Cover~\cite{hellman1970learning} addressed the problem of binary hypothesis testing within the class of time-invariant finite memory machines. 
They have studied the quantity
	\begin{align}
	    \mathsf{P^*_{e_{\rand}}}(S)=\inf_{\text{randomized} f,d} \Pe (f,d),\label{eq:Perand}
	\end{align}
	where $\Pe (f,d)$ is as defined in~\eqref{eq:pefd}, 
	and the infimum is over all time-invariant $S$-state machines with \emph{randomized} transition functions $f$. 
It was shown in~\cite{hellman1970learning} that $\mathsf{P^*_{e_{\rand}}}(S)\geq \left(1+\gamma^{\frac{S-1}{2}}\right)^{-1}$ where $\gamma=\frac{p(1-q)}{q(1-p)}$, and that this value can be approached arbitrarily closely using a randomized algorithm. 

To demonstrate the important role randomization plays in approaching this value, the same authors show in~\cite{hellman1971memory} that for any memory size $S<\infty$ and $\delta>0$ there exists problems such that any $S$-state deterministic machine has probability of error $\Pe\geq \frac{1}{2}-\delta$, while the randomized machine from~\cite{hellman1970learning} has $\Pe\leq \delta$.
When no external source of randomness is available, one can use some of the samples of $\{X_n\}$ for randomness extraction, e.g., using von Neumann extraction~\cite{von195113}. However, the extracted random bits must be stored, which could result in a substantial increase in memory~\cite{chandrasekaran1970finite}.

In~\cite{hellman1972effects} (see also~\cite{hellman1973review}) it is shown that $\underline{\E} (p,q)$, as defined in~\eqref{eq:Eunder_pq}, is positive for all $p\neq q$.\footnote{For the symmetric setting, where $p=1-q$, Shubert et al.~\cite{shubert1973testing} have also derived an upper bound on $\Pstar(S)$ that yields a positive error exponent $\underline{\E} (p,q)$.} Thus, recalling that $\mathsf{P^*_{e_{\rand}}}(S)\geq \left(1+\gamma^{\frac{S-1}{2}}\right)^{-1}$, we see that whenever $\gamma<\infty$, i.e., for any $0<p,q<1$, there exists some integer  $1\leq C=C(p,q)<\infty$ such that $\Pstar(S\cdot C)\leq \mathsf{P^*_{e_{\rand}}}(S)$,
 \emph{for all $S$}. Our main result, stated in Theorem~\ref{thrm:converse} below, may be interpreted as a lower bound on the required $C(p,q)$. Moreover, our Corollary~\ref{cor:tight} below shows that $C(p,q)$ grows unbounded for fixed $q<1/2$ and $p\to 1$. 

Finally, we note that after being abandoned for decades, the problem of learning under memory constraints is again attracting considerable attention in the machine learning literature, see, e.g.,~\cite{sd15,svw16,raz18,ds18,dks19,ssv19}. Another closely related active line of work is that of learning under communication constraints~\cite{zdjw13,bgmnw16,xr17,how18,act18,bho18}.

\section{Main Result}	
We are now ready to present our main result.
\begin{theorem}\label{thrm:converse}
Define
\begin{align}
d(p,q)&\triangleq-\frac{\log(\min \{p,1-p\})\cdot \log(\min \{q,1-q\})}{\log(\min \{p,1-p\})+\log(\min \{q,1-q\})}.\label{eq:d_pq}
\end{align}
Then 
\begin{align}
    \overline{\E} (p,q)\leq  d(p,q).
\end{align}
\end{theorem}
As it turns out, for extreme values of $p$ (resp. $q$), the bound is tight. To show that, we need the following theorem.
\begin{theorem}\label{thrm:direct}
Define
\begin{align}
    r(p,q)\triangleq \frac{\log p \log (1-q)-\log q \log (1-p)}{\log p(1-p)+\log q(1-q)}.
\end{align}
Then for every $p>q$, 
\begin{align}
   \underline{\E}(p,q) \geq r(p,q).
\end{align}
\end{theorem}
This lower bound on the error exponent is not tight in general, and in particular, for the symmetric case $p=1-q$ it is worse then the exponent derived in ~\cite{shubert1973testing}. We introduce it for the sole purpose of showing the tightness of our converse in certain regimes.
The following corollary shows that in the limit of fixed $q<\frac{1}{2}$ (resp. $p>\frac{1}{2}$) and $p \rightarrow 1$ (resp. $q \rightarrow 0$) our upper and lower bounds coincide.
\begin{corollary}\label{cor:tight}
For any fixed $q<\frac{1}{2}$,
\begin{align}
    \lim_{p \rightarrow 1}\overline{\E}(p,q)=\lim_{p \rightarrow 1}\underline{\E}(p,q)=-\log q.
\end{align}
Similarly, For any fixed $p>\frac{1}{2}$,
\begin{align}
    \lim_{q \rightarrow 0}\overline{\E}(p,q)=\lim_{q \rightarrow 0}\underline{\E}(p,q)=-\log (1-p).
\end{align}
\end{corollary}
Our converse, though in general not tight, demonstrates the gap between the error exponent for deterministic machines, and that of randomized ones, which was derived in~\cite{hellman1970learning}. Recalling that for any $q<1/2$ , the error exponent for randomized machines grows unbounded in the limit of $p\to 1$, Corollary~\ref{cor:tight} reveals that the restriction to deterministic machines may arbitrarily degrade the error exponent.

\section{Achievability} \label{sect:direct_proof} 
Before we proceed to the proof of Theorem~\ref{thrm:converse}, which is our main result, we start with upper bounding $\Pstar(S)$ by analyzing various machines. It may be instructive to review some intuitive algorithms first, in order of increasing complexity, and evaluate their respective error probabilities.
\subsection{Storing Sequences}
Assume $S$ is a power of $2$, such that $k=\log(S)$, and store $X_1,\ldots X_k$. With this strategy, the problem reduces to the standard binary hypothesis testing for which the error probability is given by $2^{-k D^*(1+o(1))}$, where $D^*$ is the Chernoff information between the two hypotheses~\cite{cover2012elements}. Therefore, the error probability is polynomially decreasing in $S$.
\subsection{Counting Ones}
The flaw in the above storage mechanism is that it wastes a tremendous amount of memory by storing all sequences, where it is sufficient to keep track of the number of ones in the sequence. 
\begin{claim}\label{lem:num_state}
Let $S^*$ be the minimal number of states required to determine whether or not a sequence of length $k$ contains at least $tk-1$ ones, for some $0<t<1$ such that $tk\in \mathbb{Z}$. Then
\begin{align}
    \frac{1}{2}\min \{t^2,(1-t)^2\}k^2 \leq S^* \leq tk^2.
\end{align}
\end{claim}
The (straightforward) proof is omitted.
From the claim we conclude that we can attain $\Pe$ that decreases exponentially in $\sqrt{S}$. 
\subsection{Proof of Theorem~\ref{thrm:direct} - Detecting Discriminating Sequences}
We begin by providing some high-level intuition guiding our construction. Since the sequence length is unbounded, one can afford to wait for the events that most sharply distinguish between the hypotheses, even if these events are arbitrarily rare. 
A reasonable choice for such events is a long consecutive run of either zeros or ones. We choose integers $a$ and $b$ such that $S=a+b+1$. If we observe a run of $a$ consecutive ones before a run of $b$ consecutive zeros we decide $\mathcal{H}_0$, and if we observe a run of $b$ consecutive zeros before a run of $a$ consecutive ones, we decide $\mathcal{H}_1$. This algorithm can be implemented using the finite-state machine with $S$ states depicted in Figure~\ref{fig:Successive runs}, for which $a=S-s$ and $b=s-1$ (the probabilities on the arrows correspond to $\mathcal{H}_0)$, where $s$ is the initial state.
\begin{center}
\begin{figure}[H]
\setlength\belowcaptionskip{-1.4\baselineskip}
\begin{tikzpicture}
  \tikzset{
    >=stealth',
    node distance=0.92cm,
    state/.style={font=\scriptsize,circle, align=center,draw,minimum size=20pt},
    dots/.style={state,draw=none}, edge/.style={->},
  }
  \node [state ,label=center:$1$] (S0)  {} ;
  \node [state] (S0-1)   [right of = S0]   {};
  \node [dots] (dots1)   [right of = S0-1]   {$\cdots$};
  \node [state] (1l) [right of = dots1]  {};
  \node [state ,label=center:$s$] (0)   [right of = 1l] {};
  \node [state] (1r) [right of = 0]   {};
  \node [dots]  (dots2)  [right of = 1r] {$\cdots$};
  \node [state] (S1-1) [right of = dots2]  {};
  \node [state ,label=center:$S$] (S1) [right of = S1-1]  {};
  \path [->,draw,thin,font=\footnotesize]  (S0-1) edge[bend left=45] node[below ] {$1-p$} (S0);
  \path [->,draw,thin,font=\footnotesize]  (0) edge[bend left=45] node[below] {$1-p$} (1l);
  \path [->,draw,thin,font=\footnotesize]  (1r) edge[bend left=45] node[below right] {$1-p$} (1l);
  \path [->,draw,thin,font=\footnotesize]  (S1-1) edge[bend left=45] node[below right] {$1-p$} (1l);
  \path [->,draw,thin,font=\footnotesize]  (1l) edge[bend left=45] node[below ] {$1-p$} (dots1);
  
  \path [->,draw,thin,font=\footnotesize]  (S0-1) edge[bend left=45] node[above left] {$p$} (1r);
  \path [->,draw,thin,font=\footnotesize]  (1l) edge[bend left=45] node[above left] {$p$} (1r);
  \path [->,draw,thin,font=\footnotesize]  (0) edge[bend left=45] node[above ] {$p$} (1r);
  \path [->,draw,thin,font=\footnotesize]  (1r) edge[bend left=45] node[above] {$p$} (dots2);
  \path [->,draw,thin,font=\footnotesize]  (S1-1) edge[bend left=45] node[above ] {$p$} (S1);
 \path [->,draw,thin,font=\footnotesize]  (S0)  edge[loop left]  node[above ]{$1$} (S0);
 \path [->,draw,thin,font=\footnotesize]  (S1)  edge[loop right] node[above ]{$1$} (S1);
\end{tikzpicture}
\caption{Counting consecutive runs of heads or tails}		\label{fig:Successive runs}
\end{figure}
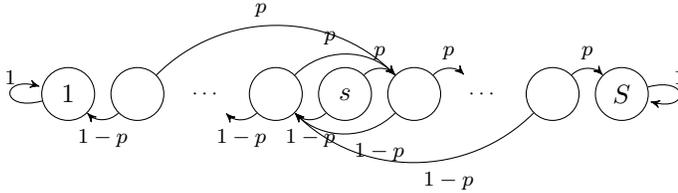
\end{center}
According to (\cite{feller1968introduction}, chapter VIII) the probability of observing a run of $a$ consecutive ones before a run of $b$ consecutive zeros under $\mathcal{H}_0$, which corresponds to the probability of absorption in state $S$ when starting in state $s$ for the machine of Figure~\ref{fig:Successive runs}, is
\begin{align}
p^0_0(s)&\triangleq \frac{1-(1-p)^b}{1+\frac{(1-p)^{b-1}}{p^{a-1}}-(1-p)^{b-1}}\\&=\frac{1-(1-p)^{s-1}}{1+\frac{(1-p)^{s-2}}{p^{S-s-1}}-(1-p)^{s-2}}.
\end{align}
Consequently, the probability of absorption in state $1$ when starting in state $s$ under the same hypothesis is
\begin{align}
p^1_0(s)&\triangleq \frac{1-p^a}{1+\frac{p^{a-1}}{(1-p)^{b-1}}-p^{a-1}}\\&=\frac{1-p^{S-s}}{1+\frac{p^{S-s-1}}{(1-p)^{s-2}}-p^{S-s-1}}.
\end{align}
Similarly, the respective probabilities under $\mathcal{H}_1$ are
\begin{align}
  p^0_1(s)&\triangleq \frac{1-(1-q)^{s-1}}{1+\frac{(1-q)^{s-2}}{q^{S-s-1}}-(1-q)^{s-2}},\\ p^1_1(s)&\triangleq \frac{1-q^{S-s}}{1+\frac{q^{S-s-1}}{(1-q)^{s-2}}-q^{S-s-1}}.
\end{align}
Since all states on the chain are transient apart from $\{1,S\}$, when $n$ is large the machine converges to one of these states with probability one.
Hence, the error probability is
\begin{align}
\Pe(s)=\frac{1}{2}(p^0_1(s)+p^1_0(s)).
\end{align}
Choosing $s=s^*$, where $s^*$ is  
\begin{align}
    \frac{\log pq}{\log p(1-p)+\log q(1-q)}S+\log \left(\frac{\frac{(1-q)^2}{q}\log q(1-q)}{\frac{p}{(1-p)^2}\log p(1-p)}\right),
\end{align}
rounded to the nearest integer, we have
\begin{align}
   \Pe(s^*)\leq &\max \left\{\frac{p^{1+c}}{(1-p)^{2-c}},\frac{(1-q)^{1+c}}{q^{2-c}}\right\} \cdot 2^{-r\left(p,q\right)  (S-1)  }
\end{align}
where $c=\log \frac{(1-p)^2(1-q)^2\log q(1-q)}{pq\log p(1-p)}$ and the result follows.

\section{Converse} \label{sect:converse_proof}	The converse of Hellman and Cover implicitly assumes that the transition probabilities between states can be as small as desired, which is true when local randomness is an unlimited resource. In deterministic machines, however, the transition probabilities can only be as small as $\min (p,1-p)$ under $\mathcal{H}_0$, or $\min (q,1-q)$ under $\mathcal{H}_1$, a fact that plays a crucial role in the proof of our converse result. We note that any finite-state machine induces a Markov chain, and proceed to prove Theorem~\ref{thrm:converse} in steps, first for ergodic Markov chains, and then for non-ergodic ones. For brevity, we denote $\Pe=\Pe(f,d)$.
\subsection{Ergodic Markov chains}
Assume the finite state machine is irreducible and aperiodic, such that the induced Markov chain is ergodic under both hypotheses. We note that, due to irreducibility, the average fraction of time spent in each state converges to a unique stationary distribution. Thus, the proof below still holds for periodic chains.

Denote by $\mu_i^p$ (resp. $\mu_i^q$) the stationary probability of state $i$ in the chain, under hypothesis $\mathcal{H}_0$ (resp. $\mathcal{H}_1$). Due to the equal prior on the hypotheses, the decision rule $d$ that minimizes~\eqref{eq:pefd} maps each state to the hypothesis with the larger stationary probability. We show that there must exist a state $i \in [S]$ for which both $\mu_i^p$ and $\mu_i^q$ are large, and that this forces $\Pe$ to be large as well.
We now proceed to formalize this idea.

\begin{lemma}\label{lem:Pmin}
Let $\{\mu_i^p\}_{i=1}^S$ be the stationary probabilities corresponding to $\Mat (p)$, and let $\{\mu_i^q\}_{i=1}^S$ be the stationary probabilities corresponding to $\Mat (q)$. Then
    \begin{align}
        \Pe \geq \frac{1}{2}\underset{i}{\max}\min \{\mu_i^p,\mu_i^q\}\triangleq \Pm\left(\{\mu_i^p\},\{\mu_i^q\}\right).
        \label{eq:maxmin_bound}
    \end{align}
\end{lemma}
\begin{proof}
 Since the prior on the hypotheses is uniform, the decision rule $d$ that minimizes~\eqref{eq:pefd} is of the form $d(i)=\ind (\mu_i^q\geq\mu_i^p)$. Hence
\begin{align}
    \Pe&=\frac{1}{2}\sum_i \mu_i^p \ind (\mu_i^q\geq\mu_i^p)+ \frac{1}{2}\sum_i \mu_i^q \ind (\mu_i^p>\mu_i^q)\\& =\frac{1}{2}\sum_i \min \{\mu_i^p,\mu_i^q\} \\&\geq \frac{1}{2}\underset{i}{\max}\min \{\mu_i^p,\mu_i^q\}.
\end{align}
\end{proof}
\begin{lemma}\label{lem:order_decrease}
Let $\{\mu_i^{\downarrow p}\}_{i=1}^S$ be an arrangement of $\{\mu_i^p\}$ in non-increasing order and $\{\mu_i^{\uparrow q}\}_{i=1}^S$ be an arrangement of $\{\mu_i^q\}$ in non-decreasing order. Then $\Pm\left(\{\mu_i^p\},\{\mu_i^q\}\right)\geq \Pm\left(\{\mu_i^{\downarrow p}\},\{\mu_i^{\uparrow q}\}\right)$. 
\end{lemma}
\begin{proof}
Since $\Pm\left(\{\mu_i^p\},\{\mu_i^q\}\right)$ is invariant to relabeling of the states, without loss of generality, we may assume $\{\mu_i^p\}=\{\mu_i^{\downarrow p}\}$. It suffices to show that if $\mu_j^q\leq \mu_i^q$ for $j>i$, then swapping $\mu_j^q$ with $\mu_i^q$ cannot increase the maxmin in~\eqref{eq:maxmin_bound}. Let $j>i$ and let $\left(\mu_i^p,\mu_i^q\right)=(a,c)$, $\left(\mu_j^p,\mu_j^q\right)=(b,d)$, where $a\geq b,c \geq d$. The restriction of the maxmin to the nodes $(i,j)$ is given by 
\begin{align}
 \max \left(\min\{a,c\},\min \{b,d\}\right)\geq \min\{a,c\} .  
\end{align}
Replacing $\mu_j^q$ with $\mu_i^q$ changes this value to 
\begin{align}
\max \left(\min\{a,d\},\min \{b,c\}\right) &\leq \max \left(\min\{a,c\},\min \{a,c\}\right)\nonumber\\&= \min\{a,c\},  
\end{align}
which clearly cannot increase the maxmin.
\end{proof}

The next lemma exploits the restriction to deterministic machines.
\begin{lemma}\label{lem:order_increase}
Let $\{\mu_i^{\downarrow p}\}_{i=1}^S$ be an arrangement of $\{\mu_i^p\}$ in non-increasing order. Then:
\begin{align}
    \mu _{i+1}^{\downarrow p}\geq \mu_i^{\downarrow p}\cdot \min \{p,1-p\}.
\end{align}
\end{lemma}
\begin{proof}
Without loss of generality, we may relabel the states such that $\mu_i^{\downarrow p}=\mu_i^p$, for all $i$. Let $A=\{1,\ldots,i\}$ and consider the partition of $S$ to $S = A\cup A^c$. Since the chain is irreducible, there is some $j\in A^c$ that is accessible from some $j'\in A$ in one step. Then
\begin{align}
\mu _{i+1}^{\downarrow p}\geq \mu _j^{\downarrow p}&\geq \mu^{\downarrow p}_{j'}\cdot \min \{p,1-p\}\\&\geq \mu_i^{\downarrow p}\cdot \min \{p,1-p\}.  
\end{align}
\end{proof}
\subsection*{Proof of Theorem~\ref{thrm:converse} for ergodic Markov chains:}
A repeated application of Lemma~\ref{lem:order_increase} implies that 
\begin{align}
 \mu _{i}^{\downarrow p}&\geq\mu _{1}^{\downarrow p}\min \{p,1-p\}^{i-1}\\&\geq \frac{1}{S}\min \{p,1-p\}^{i-1},   
\end{align}
as well as 
\begin{align}
 \mu _{i}^{\uparrow q}&\geq \mu _{S}^{\uparrow q}\min \{q,1-q\}^{S-i}\\&\geq \frac{1}{S}\min \{q,1-q\}^{S-i},   
\end{align}
where we used the fact that the largest stationary probability among all states must be at least $\frac{1}{S}$. From Lemma~\ref{lem:Pmin} and Lemma~\ref{lem:order_decrease}, by ordering $\mu_i^p$ in decreasing order and $\mu_i^q$ in increasing order, we get the following lower bound on the error probability, 
\begin{align}
    \Pe\geq \frac{1}{S}\cdot \underset{i}{\max}\min \left\{\min \{p,1-p\}^{i-1},\min \{q,1-q\}^{S-i}\right\}.\label{eq:ergodic_pe}
\end{align}
Since both functions are monotone in $1\leq i\leq S$, one is decreasing from $1$ and the other is increasing to $1$, the maximum over $i\in [1,S]$ is attained for $i$ such that $\min \{p,1-p\}^{i-1}=\min \{q,1-q\}^{S-i}$, namely, for
\begin{align}
   i=\frac{\log\min\{q,1-q\}}{\log\left(\min \{p,1-p\}\min \{q,1-q\}\right)} S +\log\min \{p,1-p\}.\label{eq:optimal_i}
\end{align}
As $i$ must be an integer, the expression above should be rounded up or down. However, asymptotically this has no effect on the bound. Substituting~\eqref{eq:optimal_i} into~\eqref{eq:ergodic_pe}, the theorem follows for the ergodic case.
\subsection{Non-Ergodic Markov chains}
Consider first the case where we have only two absorbing states, one for each hypothesis, i.e., assume that we decide  $\mathcal{H}_0$ if the process is absorbed in state $S$ and  $\mathcal{H}_1$ if the process is absorbed in state $1$ . Define $X_0$ and $X_1$ as the independent random walks under $\mathcal{H}_0$ and $\mathcal{H}_1$. Then $X_0$ (resp. $X_1$) is a stochastic process over the alphabet $[S]$ that starts at $s$ and evolves according to the stochastic matrix $\Mat(p)$ (resp. $\Mat(q)$). Define the conditional error probabilities:
\begin{align}
   p_0 &= \Pr( 1\in X_0) ,\label{eq:p0ex}
   \\p_1 &= \Pr(S\in X_1),\label{eq:p1ex}
\end{align}
and hence $\Pe=\frac{1}{2}(p_0+p_1)$. Define the total distance of a state $u$ to be the smallest sum of lengths of two simple paths from $u$ to $1$ and from $u$ to $S$, and denote it by $\td(u)$. Furthermore, define the occupancy of a state $u$ to be the minimal probability that one of the random walks will visit it, i.e., $\occ(u)\triangleq \min_i \Pr(u^* \in X_i)$. 
A simple bound on the error probability of any system is the probability of the shortest path from $s$ to the incorrect absorbing state under either hypothesis. However, such a bound may not be tight, since $s$ itself can only be guaranteed to have $\td(s)\leq 2S$. To see this, consider that the shortest path to each state cannot be larger than $S$, and is exactly $S$ for the linear graph that splits at the last node to either absorbing state. On the other hand, the best possible guarantee we can hope for is total distance of $S$, which corresponds to a chain in which the shortest paths are non-intersecting.  
This motivates us to find a state with the smallest possible total distance and a non-negligible occupancy.
\begin{lemma}\label{lem:walk_to_sink}
There exists a state $u^*$ with $\td(u^*)\leq S$ and 
\begin{align}
    \occ(u^*) \geq \frac{1- \max\{p_0,p_1\}}{S},
\end{align}
where $p_0$ and $p_1$ are as in~\eqref{eq:p0ex},~\eqref{eq:p1ex}.
\end{lemma}
\begin{proof}
Let $\mathcal{A}$ (resp. $\mathcal{B}$) denote the collection of all simple paths that start at $s$ and terminate at $1$ (resp. $S$). Let $\mathcal{C}$ be the set of all vertices $v\in [S]$, for which there exist two simple paths $a \in \mathcal{A}$ and $b \in \mathcal{B}$, where $v$ is the last vertex in $a$ that also appears in $b$. 
This implies that the sum of path lengths from any $v\in \mathcal{C}$ to $1$ and $S$ is smaller than $S$, i.e., $\forall v \in \mathcal{C}$ we have $\td(v)\leq S$. Define $\tilde{X}_0$ (resp. $\tilde{X}_1$) to be a stochastic process with the distribution of $X_0$ (resp. $X_1$) conditioned on the event that $X_0$ terminated at $S$ (resp. $1$). Define $U$ to be the last state on $\tilde{X}_0$ that also appears on $\tilde{X}_1$. Then by definition $\Pr(U\in \mathcal{C}) = 1 $, so there must be a state $u^* \in S$ such that 
\begin{align}
 \Pr(U=u^*)\geq \frac{1}{|\mathcal{C}|}\geq \frac{1}{S}.  
\end{align}
This in particular implies that  $\Pr(u^* \in \tilde{X}_0)\geq \frac{1}{S}$ and $\Pr(u^* \in \tilde{X}_1)\geq \frac{1}{S}$. Now, the probability of the unconditioned walk $X_1$, to pass through $u^*$ is lower bounded by 
\begin{align}
    \Pr\left(u^* \in X_1\right) &\geq \Pr(1\in X_1) \Pr(u^* \in X_1|1\in X_1)\\&=\Pr(1\in X_1) \Pr(u^* \in \tilde{X_1})\\&\geq \frac{1}{S} \left(1-p_1\right).
\end{align}
Similarly bounding $\Pr\left(u^* \in X_0\right)$, the lemma follows. 
\end{proof} 
\subsection*{Proof of Theorem~\ref{thrm:converse} for two absorbing states:}
Without loss of generality, we may assume that $\max\{p_0,p_1\}  < 1/2$ as otherwise the theorem is trivially true. Furthermore, from Lemma~\ref{lem:walk_to_sink} there is some state $u^*$ with $\td(u^*)\leq S$ and $\occ(u^*) \geq \frac{1- \max\{p_0,p_1\}}{S}$. Write
\begin{align}
\Pe &\geq \frac{1}{2}\Pr(u^* \in X_0)\Pr(1 \in X_0|u^* \in X_0) \\&+\frac{1}{2}\Pr(u^* \in X_1)\Pr(S \in X_1|u^* \in X_1)\\ &\geq \frac{1- \max\{p_0,p_1\}}{2S}(\Pr(1 \in X_0|u^* \in X_0) \nonumber \\&\hspace{20mm}+\Pr(S \in X_1|u^* \in X_1)).
\end{align}
Let $m_{u^*}$ be the length of the shortest path from $u^*$ to $1$, and recall that we must have a path from $u^*$ to $S$ of length smaller than $S-m_{u^*}$, since that $\td(u^*)\leq S$. Thus,
\begin{align}
    &\Pr(1 \in X_0|u^* \in X_0)+\Pr(S \in X_1|u^* \in X_1)&\\\geq &\left(\min\{p,1-p\}\right) ^{m_{u^*}}+\left(\min\{q,1-q\}\right) ^{S-m_{u^*}}\label{eq:error_for_mu}. 
\end{align}
Minimizing the lower bound with respect to $m_{u^*}\in [0,S]$ yields
\begin{align}
 m_{u^*}=\frac{\log\min \{q,1-q\}}{\log\min \{p,1-p\}+\log\min \{q,1-q\}}\cdot S \label{eq:optimal_mu},  
\end{align}
and substituting~\eqref{eq:optimal_mu} into~\eqref{eq:error_for_mu}  implies the theorem for the case of two absorbing states.

\subsection*{Proof of Theorem~\ref{thrm:converse} for the general reducible case:}
Consider a Markov chain with $K$ recurrent classes $\mathcal{R}_1,\ldots,\mathcal{R}_{K}$, and a set $\mathcal{T}$ of transient states with initial state $s$. Note that if $s\notin \mathcal{T}$ the chain is essentially an ergodic one, hence we consider only $s\in \mathcal{T}$.
Define $X_0$ and $X_1$ as before, and denote the probability that $X_i$ ends up in class $\mathcal{R}_j$ as
\begin{align}
 \Pr \left(X_i\rightarrow \mathcal{R}_j\right),\hspace{1mm} i=0,1,\hspace{1mm} j=1,\ldots,K .
\end{align}
We further denote the  probability of error under hypothesis $\mathcal{H}_i$ if the initial state were in class $\mathcal{R}_j$ as $\Pe(\mathcal{R}_j|\mathcal{H}_i)$.
Consider first the case where the probability of error under $\mathcal{H}_0$ is larger than the probability of error under $\mathcal{H}_1$ in every recurrent class. Then
\begin{align}
  \Pe &\geq \frac{1}{2}\min_{1\leq j\leq K} \Pe(\mathcal{R}_j|\mathcal{H}_0)\label{eq:bound_H0} \\&\geq \frac{1}{2}\cdot 2^{-\underset{1\leq j\leq K}{\max}|\mathcal{R}_j|\cdot (d(p,q)+o(1))} \label{eq:one class_bound}\\&\geq 2^{-S\cdot (d(p,q)+o(1))} ,
\end{align}
where $d(p,q)$ was defined in~\eqref{eq:d_pq} and $o(1)$ is relative to $S$. Note that in~\eqref{eq:bound_H0} we bound the error probability under $\mathcal{H}_0$ with the smallest error probability across classes, and in~\eqref{eq:one class_bound} we used the fact that the error probability under $\mathcal{H}_0$ in class $\mathcal{R}_j$ is larger than the average error probability, and then used Theorem~\ref{thrm:converse} for the ergodic case. 

For the second case, we define the non-empty sets
\begin{align}
\mathcal{C}_1&=\{\mathcal{R}_k : \Pe(\mathcal{R}_k|\mathcal{H}_0)\geq \Pe(\mathcal{R}_k|\mathcal{H}_1)\},\\
\mathcal{C}_0&=\{\mathcal{R}_k : \Pe(\mathcal{R}_k|\mathcal{H}_0)<\Pe(\mathcal{R}_k|\mathcal{H}_1)\}.   
\end{align}
For any $k\in \mathcal{C}_1$, we have
\begin{align}
 \Pe(\mathcal{R}_k|\mathcal{H}_0)\geq 2^{-|\mathcal{R}_k| \cdot (d(p,q)+o(1))},  
\end{align}
and for any $k\in \mathcal{C}_0$ we have
\begin{align}
\Pe(\mathcal{R}_k|\mathcal{H}_1) \geq 2^{-|\mathcal{R}_k| \cdot (d(p,q)+o(1))},  
\end{align} 
according to Theorem~\ref{thrm:converse} for the ergodic case. Now, write
\begin{align}
 \Pe \geq \frac{1}{2} & \Pr(X_0\rightarrow \mathcal{C}_1)\min_{k\in \mathcal{C}_1}\Pe(\mathcal{R}_k|\mathcal{H}_0)\\+ \frac{1}{2}&\Pr(X_1\rightarrow \mathcal{C}_0)\min_{k\in \mathcal{C}_0}\Pe(\mathcal{R}_k|\mathcal{H}_1)
\\ \geq \frac{1}{2}&\left(\Pr(X_0\rightarrow \mathcal{C}_1)+\Pr(X_1\rightarrow \mathcal{C}_0)\right)\nonumber\\ \times
&2^{-\max\{\underset{k\in \mathcal{C}_1}{\max}|\mathcal{R}_k|,\underset{k'\in \mathcal{C}_0}{\max}|\mathcal{R}_{k'}|\}\cdot (d(p,q)+o(1))}
\\=\frac{1}{2}&\left(\Pr(X_0 \rightarrow \mathcal{C}_1)+\Pr(X_1 \rightarrow \mathcal{C}_0)\right)\nonumber \\\times &2^{-\underset{k}{\max}|\mathcal{R}_k|\cdot (d(p,q)+o(1))}.  \label{eqref:p0_condition}
\end{align}
Consider a chain with $|\mathcal{T}|+2$ states, obtained from the original chain by merging the states in $\mathcal{C}_0$ and $\mathcal{C}_1$ into two respectively absorbing states. 
Then Lemma~\ref{lem:walk_to_sink} holds, with
\begin{align}
p_0&=\Pr(X_0 \rightarrow \mathcal{C}_1),
\\p_1&=\Pr(X_1 \rightarrow \mathcal{C}_0).  
\end{align}
According to~\eqref{eqref:p0_condition}, we may assume that $\max_i p_i < 1/2$ as otherwise the theorem is trivially true. Now, repeating the same arguments as in the proof of the two absorbing states, one can show that 
\begin{align}
  &\Pr(X_0\rightarrow \mathcal{C}_1)+\Pr(X_1\rightarrow \mathcal{C}_0)\\ \geq & \frac{1-\max_i p_i}{|\mathcal{T}|+2}\cdot 2^{-(|\mathcal{T}|+2)\cdot (d(p,q)+o(1))}.
\end{align}
The proof follows by noting that $\underset{k}{\max}|R_k|+|\mathcal{T}|\leq S-1$.
\bibliography{binary_hypothesis_testing}
\bibliographystyle{ieeetr}
\end{document}